\documentclass[11pt]{article}

\usepackage[utf8]{inputenc}
\usepackage[russian,english]{babel}
\usepackage{amsmath}
\usepackage{amssymb}
\usepackage{amsthm}
\usepackage{fullpage}
\usepackage{mathpazo}

\numberwithin{equation}{section}
\usepackage[pagebackref,letterpaper=true,colorlinks=true,pdfpagemode=none,
linkcolor=blue,citecolor=blue,pdfstartview=FitH]{hyperref}
\usepackage[nameinlink,capitalize]{cleveref}

\newtheorem{theorem}{Theorem}[section]
\newtheorem{corollary}[theorem]{Corollary}
\newtheorem{lemma}[theorem]{Lemma}

\newtheorem{definition}[theorem]{Definition}

\newtheorem{remark}[theorem]{Remark}

\newcommand{\prob}[2]{\mathop{\mathrm{Pr}}_{#1}[#2]}
\newcommand{\avg}[2]{\mathop{\textbf{E}}_{#1}[#2]}

\newcommand{\AC}{\mbox{\rm AC}}

\newcommand{\integer}{\mathbb{Z}}

\newcommand{\Mod}[1]{\ (\text{mod}\ #1)}

\newcommand{\mc}[1]{\mathcal{#1}}

\newcommand{\maj}{\mathrm{Maj}}
\newcommand{\F}{\mathbb{F}}

\newcommand{\agr}{\mathrm{agr}}

\newcommand{\Maj}{\mathrm{Maj}}
\newcommand{\Ringm}[1]{\mathbb{Z}/#1 \mathbb{Z}}
\newcommand{\ACzeroparity}{{\AC^0[\oplus]}}
\newcommand{\MOD}{\mbox{\rm MOD}}

\newcommand{\srikanth}[1]{}
\newcommand{\ph}[1]{}

\author{Abhishek Bhrushundi\thanks{Department of Computer Science, Rutgers University, USA. {\tt
      abhishek.bhr@cs.rutgers.edu}. Work done while the author was visiting the Tata Institute of Fundamental Research. Research supported in part by UGC-ISF grant 6-2/2014(IC).}
\and 
   Prahladh Harsha\thanks{Tata Institute of Fundamental Research, India. {\tt
       prahladh@tifr.res.in}. Research supported in part by UGC-ISF grant 6-2/2014(IC).}
\and
Srikanth Srinivasan\thanks{Department of Mathematics, Indian Institute of Technology, Bombay, India.
  {\tt srikanth@math.iitb.ac.in}.}
}

\title{On polynomial approximations over $\Ringm{2^k}$\thanks{A
    preliminary version of this paper appeared in  {\em Proc.\
  $34$th Annual Symp.\ on Theoretical Aspects of Comp.\ Science (STACS)}, 2017~\cite{BhrushundiHS2017}.}}
\date{}

\begin{document}

\begin{titlepage}

\maketitle

\thispagestyle{empty}
\setcounter{page}{0}

\begin{abstract}
We study approximation of Boolean functions by low-degree polynomials over the ring $\Ringm{2^k}$. More precisely, given a Boolean function $F:\{0,1\}^n \rightarrow \{0,1\}$, define its $k$-lift to be $F_k:\{0,1\}^n \rightarrow \{0,2^{k-1}\}$ by $F_k(x) = 2^{k-F(x)} \pmod {2^k}$. We consider the fractional agreement (which we refer to as $\gamma_{d,k}(F)$) of $F_k$ with degree $d$ polynomials from $\Ringm{2^k}[x_1,\ldots,x_n]$. 

Our results are the following:
\begin{itemize}
\item Increasing $k$ can help: We observe that as $k$ increases, $\gamma_{d,k}(F)$ cannot decrease. We give two kinds of examples where $\gamma_{d,k}(F)$ actually increases. The first is an infinite family of functions $F$ such that $\gamma_{2d,2}(F) - \gamma_{3d-1,1}(F) \geq \Omega(1)$. The second is an infinite family of functions $F$ such that $\gamma_{d,1}(F)\leq\frac{1}{2}+o(1)$ --- as small as possible --- but $\gamma_{d,3}(F) \geq \frac{1}{2}+\Omega(1)$.

\item Increasing $k$ doesn't always help: Adapting a proof of Green~[{\em Comput.\ Complexity}, 9(1):16--38, 2000], we show that irrespective of the value of $k$, the Majority function $\Maj_n$ satisfies $$\gamma_{d,k}(\Maj_n) \leq \frac{1}{2}+\frac{O(d)}{\sqrt{n}}.$$ In other words, polynomials over $\Ringm{2^k}$ for large $k$ do not approximate the majority function any better than polynomials over $\Ringm{2}$.
\end{itemize}

We observe that the model we study subsumes the model of \emph{non-classical polynomials} 
in the sense that proving bounds in our model implies bounds on the agreement of non-classical polynomials with Boolean functions. In particular, our results answer questions raised by Bhowmick and Lovett~[{\em In Proc. 30th Computational Complexity Conf., pages 72–-87, 2015}] that ask whether non-classical polynomials approximate Boolean functions better than classical polynomials of the same degree.
\end{abstract}
\end{titlepage}

\section{Introduction}

Many lower bound results in circuit complexity are proved by showing that any small sized circuit in a given circuit class can be approximated by a function from a simple computational model (e.g., small depth circuits by low-degree polynomials) and subsequently showing that this is not possible for some suitable ``hard function''.

A classic case in point is the work of Razborov~\cite{Razborov1987} which shows lower bounds for $\ACzeroparity$, the class of constant depth circuits made up of AND, OR and $\oplus$ gates. Razborov shows that any small $\ACzeroparity$ circuit C can be well approximated by a low-degree multivariate polynomial $Q(x_1,\dots,x_n) \in \F_2[x_1,\dots,x_n]$ in the sense that $$\Pr_{x \sim \{0,1\}^n}\left[Q(x) \neq C(x) \right] = o(1).$$ The next step in the proof is to show that the hard function, on the other hand, does not have any such approximation. Razborov does this for a suitable symmetric function, Smolensky~\cite{Smolensky1987} for the $\MOD_q$ function (for constant odd $q$), and Szegedy~\cite{Szegedy1989} and Smolensky~\cite{Smolensky1993} for the Majority function $\Maj_n$ on $n$ bits.

Given the importance of the above lower bound, polynomial approximations in other domains and metrics have been intensely investigated and have resulted in interesting combinatorial constructions and error-correcting codes~\cite{Grolmusz2000,Efremenko2012}, learning algorithms~\cite{LinialMN1993,KlivansS2004} and more recently in the design of algorithms for combinatorial problems~\cite{Williams2014,AbboudWY2015} as well.

To describe the model of polynomial approximation considered in this paper, we first recall the Razborov~\cite{Razborov1987} model of polynomial approximation. Given a Boolean function $F:\{0,1\}^n \to \{0,1\}$ and degree $d \leq n$, Razborov considers the largest $\gamma$ such that there is a degree $d$ polynomial $Q \in \F_2[x_1,\dots,x_n]$ that has agreement at least $\gamma$ with $F$ (i.e., $\Pr_x[Q(x) = F(x) ] \geq \gamma$). Call this $\gamma_d(F)$. In this notation, Szegedy~\cite{Szegedy1989} and Smolensky's~\cite{Smolensky1993} results for the Majority function can be succinctly stated as $$\gamma_d(\Maj_n) \leq \frac12+\frac{O(d)}{\sqrt{n}}.$$ We consider a generalization of the above model to rings $\Ringm{2^k}$ in the following simple manner. To begin with, we consider the ring $\Ringm{4}$. Given a Boolean function $F$, let $F_2 : \{0,1\}^n \to \{0,2\} \subseteq \Ringm{4}$ be the $2$-lift of $F$ defined as $F_2(x) := 2^{2-F(x)}$ (i.e., $F_2(x) := 0 $ if $F(x) = 0$ and $F_2(x) := 2$ otherwise). Once again, we can define $\gamma_{d,2}(F)$ to be the largest $\gamma$ such that there exists a degree $d$ polynomial $Q_2 \in \Ringm{4}[x_1,\dots,x_n]$ that has agreement $\gamma $ with $F_2$. Note that $\gamma_{d,2}(F) \geq \gamma_d(F)$ since if, for instance, $Q(x) = x_1x_2+x_3 \in \F_2[x_1,\dots,x_n]$ has agreement $\gamma $ with $F$, then $Q_2: = 2(x_1x_2+x_3) \in \Ringm{4}[x_1,\dots,x_n]$ also has the same agreement $\gamma$ with $F_2$. Hence, proving upper bounds for $\gamma_{d,2}(F)$ is at least as hard as proving upper bounds for $\gamma_d(F)$.

More generally, we can extend these definitions to $\gamma_{d,k}(F)$, the agreement of $F_k$, the $k$-lift of $F$, defined as $F_k(x) = 2^{k-F(x)} \mod {2^k}$, with degree $d$ polynomials from $\Ringm{2^k}[x_1,\dots,x_n]$. It is not hard to show that $\gamma_{d,k+1}(F) \geq \gamma_{d,k}(F)$ and hence as $k$ increases, the problem of proving upper bounds on $\gamma_{d,k}(F)$ can only get harder.

Our motivation for this model comes from a recent work of Bhowmick and Lovett~\cite{BhowmickL2015}, who study the maximum agreement between \emph{non-classical polynomials} of degree $d$ and a Boolean function $F$, which is similar to $\gamma_{d,d}(F)$ (see \cref{sec:nonclass} for an exact translation between the above model and non-classical polynomials). In particular, non-classical polynomials of degree $d$ can be considered as a subset of the degree $d$ polynomials in $\Ringm{2^d}[x_1,\dots,x_n]$. With respect to correlation\footnote{The correlation between $F,G:\{0,1\}^n\rightarrow \Ringm{2^k}$ is defined to be $\avg{x}{\omega^{F(x)-G(x)}}$ where $\omega$ is the primitive $2^k$th root of unity in $\mathbb{C}$. If $F,G$ are $\{0,2^{k-1}\}$-valued, then this quantity is exactly $2\gamma-1$ where $\gamma$ is the agreement between $F$ and $G$. Otherwise, however, it does not measure agreement.}, Bhowmick and Lovett showed that there exist non-classical polynomials (and hence polynomials in $\Ringm{2^d}[x_1,
\dots,x_n]$) of logarithmic degree that have very good correlation with the $\Maj_n$ function. With respect to agreement, they show that low-degree non-classical polynomials can only have \emph{small} agreement with the Majority function. Their results stated in our language, imply that 
$$
\label{eq:BLineq}
\gamma_{d,d}(\Maj_n) \leq \frac12 + \frac{O(d\cdot2^d)}{\sqrt{n}}.
$$
In particular, if $d = \Omega(\log n)$, this result unfortunately does not give any non-trivial bound on the maximum agreement between non-classical polynomials of degree $d$ and the $\Maj_n$ function. Bhowmick and Lovett, however, conjectured that this result could be improved and left open the question of whether non-classical polynomials of degree $d$ can do any better than classical polynomials of the same degree in approximating the Majority function. More generally, they informally conjectured that although non-classical polynomials achieve better correlation with Boolean functions than their classical counterparts, they possibly do not approximate Boolean functions any better than classical polynomials. Our work stems from trying to answer these questions.

\subsection{Our results}
We prove the following results about agreement of Boolean functions with polynomials over the ring $\Ringm{2^k}$:

\begin{enumerate}
\item We explore whether there exist Boolean functions for which agreement can increase by increasing $k$. In particular, do there exist Boolean $F$ such that $\gamma_{d,k}(F) > \gamma_{d,1}(F)$?

It is not hard to show that this is impossible for $d=1$. Further, it can be shown that if $\gamma_{d,k}(F) > 1-\frac{1}{2^d}$, then $\gamma_{d,k(}F) = \gamma_{d,1}(F)$. Keeping this in mind, the first place where we can expect larger $k$ to show better agreement is $\gamma_{2,2}$ vs. $\gamma_{2,1}$.
Our first result shows that there are indeed separating examples in the regime.
\begin{enumerate}
\item[(a)] Fix  $d \in \mathbb{N}$ to be any power of $2$. For infinitely many $n$, there exists a Boolean function $F: \{0,1\}^n \to \{0,1\}$ such that $\gamma_{3d-1,1}(F) \leq 5/8 + o(1)$ but $\gamma_{2d,2}(F) \geq 3/4-o(1)$.
\end{enumerate}
Note that since $F$ is Boolean, $\gamma_{d,k}(F) \geq 1/2$ for any $d,k$. We then ask if there exist Boolean functions $F$ such that $\gamma_{d,1}(F)$ is more or less the trivial bound of $1/2$, while $\gamma_{d',k}(F)$ is significantly larger for $d' \le d$ and some $k > 1$. In this context, we show the following result.
\begin{enumerate}
\item[(b)] Fix any $\ell \ge 2$. For large enough $n$, there is a Boolean function $F:\{0,1\}^n \rightarrow \{0,1\}$ such that $\gamma_{2^{\ell}-1,1}(F) \leq 1/2 +o(1)$ but $\gamma_{d,3}(F) \geq 9/16 -o(1)$, for $d = 2^{\ell-1}+2^{\ell-2} \le 2^{\ell} - 1$. 
\end{enumerate}

\item We show that for $\Maj_n$, the majority function on $n$ bits, and any $d, k \in \integer^+$, $$\gamma_{d,k}(\Maj_n) \leq \frac12 + \frac{O(d)}{\sqrt{n}},$$\footnote{The constant in the $O(\cdot)$ is an absolute constant.} by adapting a proof due to Green~\cite{Green2000} of a result on the approximability of the parity function by low-degree polynomials over the ring $\Ringm{p^k}$ for prime $p \neq 2$. 

\end{enumerate}
Coupled with the observation that the class of polynomials over rings $\Ringm{2^k}$ subsumes the class of non-classical polynomials, part $(b)$ of the first result provides a counterexample to an informal conjecture of Bhowmick and Lovett~\cite{BhowmickL2015} that, for any Boolean function $F$, non-classical polynomials of degree $d$ do not approximate $F$ any better than classical polynomials of the same degree, and the second result confirms their conjecture that non-classical polynomials do not approximate the Majority function any better than classical polynomials.  

\subsection{Organisation}
We start with some preliminaries in \cref{sec:prelims}.  In \cref{sec:examples}, we show some separation results.  Next, in \cref{sec:smol}, we prove upper bounds for $\gamma_{d,k}(\Maj_n)$. Finally, in \cref{sec:nonclass}, we discuss how our model relates to non-classical polynomials, answering questions raised by Bhowmick and Lovett.
\section{Preliminaries}
\label{sec:prelims}

For $x \in \{0,1\}^n$, $|x|$ denotes the Hamming weight of $x$, and for $i \ge 0$, $|x|_i$ is the $(i+1)^{\mbox{th}}$ least significant bit of $|x|$ in base $2$. For $d\in \mathbb{N}$, we use $\{0,1\}^n_{\leq d}$ (resp. $\{0,1\}^n_{=d}$) to denote the set of elements in $\{0,1\}^n$ of Hamming weight at most $d$ (resp. exactly $d$). We use $\mc{F}_n$ to denote the collection of all Boolean functions defined on $\{0,1\}^n$.

\subsection{Elementary symmetric polynomials}
Recall that for $t \ge 1$, the elementary symmetric polynomial of degree $t$ over $\mathbb{F}_2$, $S_t(x_1, \ldots, x_n)$, is defined as $S_t(x_1,\ldots,x_n) = \bigoplus_{1 \le a_1<\ldots<a_t \le n} x_{a_1}\ldots x_{a_t}$.
Here $\oplus$ denotes addition modulo two. This may be interpreted as
\begin{equation}
\label{eqn:symdef2}
S_t(x_1,\ldots,x_n) = \binom{|x|}{t} \mbox{ mod } 2.
\end{equation}

A direct consequence of Lucas's theorem (see, e.g.,~\cite[Section 1.2.6, Ex. 10]{Knuth-tacopI}) and~\cref{eqn:symdef2} is the following:
\begin{lemma}
 \label{fac:elemsym1}
 For every $\ell \ge 0$,  $S_{2^\ell}(x) = |x|_\ell$. More generally, $S_t(x) = \prod_{i} |x|_i$ where  the product runs over all $i\geq 0$ such that the $(i+1)^{\mbox{th}}$ least significant bit of the binary expansion of $t$ is $1$.
\end{lemma}

The following result follows from the work of Green and Tao~\cite[Theorem 11.3]{GreenT2009}, who build upon the ideas of Alon and Beigel ~\cite{AlonB2001}.
\begin{theorem}[Green-Tao~\cite{GreenT2009}, Alon-Beigel~\cite{AlonB2001}]
\label{thm:alonbeig}
Fix $\ell \ge 0$. Then, for every multilinear polynomial $P \in \F_2[x_1,\ldots,x_n]$ of degree at most $2^\ell -1$, we have $\Pr_{x \sim \{0,1\}^n}[S_{2^\ell}(x) = P(x)] \le 1/2 + o(1)$.
\end{theorem}

\cref{thm:alonbeig} has a nice corollary:
\begin{corollary}
\label{cor:folk}
 For every fixed $\ell \ge 0$, the functions $\{S_{2^i}(x)\}_{0\le i \le \ell}$ are almost balanced and almost uncorrelated, i.e.
 \begin{itemize}
  \item $\forall\ 0\le i \le \ell$, $|\Pr[S_{2^i}(x) = 0] - \Pr[S_{2^i}(x) = 1]| = o(1)$
  \item $\forall\ a_0, \ldots, a_\ell \in \{0,1\}$, $|\Pr\left[\bigwedge_{0\le i \le \ell} \left(S_{2^i}(x) = a_i\right)\right] - \frac{1}{2^{\ell+1}}| = o(1)$.
 \end{itemize}
\end{corollary}

Combining~\cref{cor:folk} with~\cref{fac:elemsym1}, we get another useful fact:
\begin{lemma}
\label{fac:folk}
Let $x$ be uniformly distributed over $\{0,1\}^n$. Then, for every fixed $r \ge 1$, the random variables $\{|x|_i\}_{0\le i \le r-1}$ are almost uniform and almost $r$-wise independent i.e.
\begin{itemize}
 \item $\forall\ 0 \le i \le r-1$, $|\Pr[|x|_i = 0] - \Pr[|x|_i = 1]| = o(1)$.
 \item $\forall\ (a_0,\ldots,a_{r-1}) \in \{0,1\}^{r}$, $| \Pr[(|x|_0,\ldots,|x|_{r-1}) = (a_0,\ldots,a_{r-1})] - \frac{1}{2^{r}} | = o(1)$. 
\end{itemize} 
\end{lemma}

\subsection{Boolean functions and polynomials over $\Ringm{2^k}$}

Given an $F\in \mc{F}_n$ and $k \ge 1$, we define the \emph{$k$-lift of $F$} to be the function $F_k:\{0,1\}^n \rightarrow \Ringm{2^{k}}$ defined as follows. For any $x\in \{0,1\}^n$,
\[
F_k(x) = \left\{
\begin{array}{cc}
0 & \text{if $F(x) = 0$,}\\
2^{k-1} & \text{otherwise.}
\end{array}\right.
\]

For $d\in \mathbb{N}$ and $k \ge 1$, $\mc{P}_{d,k}$ will denote the set of multilinear polynomials of degree at most $d$ over the ring $\mathbb{Z}/2^k\mathbb{Z}$. 

For functions $F,G:D \rightarrow R$ for some finite domain $D$ and range $R$, the \emph{agreement} between $F$ and $G$, denoted by $\agr(F,G)$, is defined to be the fraction of inputs where they agree, i.e., 
\[
\agr(F,G) = \prob{x\sim D}{F(x) = G(x)}.
\]

We will consider how well multilinear polynomials of degree $d$ can approximate Boolean functions in the above sense. More precisely, for any Boolean function $F\in \mc{F}_n$, we define 
\[
\gamma_{d,k}(F) = \max_{Q\in \mc{P}_{d,k}} \agr(F_k,Q).
\]

Following~\cite{Gopalan2008}, we call a set $I\subseteq \{0,1\}^n$ an \emph{interpolating set}\footnote{This is also called a \emph{hitting set} in the literature.} for $\mc{P}_{d,k}$ if the only polynomial $P\in \mc{P}_{d,k}$ that vanishes at all points in $I$ is zero everywhere. Formally, for any $P\in \mc{P}_{d,k}$,
\[
(\forall x\in I\ \  P(x) = 0) \Rightarrow (\forall y \in \{0,1\}^n\ \  P(y) = 0).
\]

We now state a number of standard facts regarding Boolean functions and multilinear polynomials over $\Ringm{2^k}$. The omitted proofs are either easy or well-known.

Unless mentioned otherwise, let $n,d,k$ be any integers satisfying $n\geq 1, d\geq 0, k\geq 1$.

\begin{lemma}
\label{fac:polys}
Any polynomial $Q\in \mc{P}_{d,k}$ satisfies the following:
\begin{enumerate}
\item (Schwartz-Zippel) If $Q$ is non-zero, then $\prob{x\sim \{0,1\}^n}{Q(x) \neq 0} \geq \frac{1}{2^d}$. 
\item $Q$ is the zero polynomial iff $Q(x) =0$ for all $x\in \{0,1\}^n$.
\item (M\"{o}bius Inversion) Say $Q(x) = \sum_{|S|\leq d}c_S x_S$, where $c_S\in \Ringm{2^k}$ and $x_S$ denotes $\prod_{i\in S}x_i$. Then, $c_S = \sum_{T\subseteq S}(-1)^{|S|-|T|}Q(1_T)$ where $1_T\in \{0,1\}^n$ is the characteristic vector of $T$.
\item ($\{0,1\}^n_{\leq d}$ is an interpolating set) $Q$ vanishes at all points in $\{0,1\}^n$ iff $Q$ vanishes at all points of $\{0,1\}^n_{\leq d}$. By shifting the origin to any point of $\{0,1\}^n$, the same is true of \emph{any} Hamming ball of radius $d$ in $\{0,1\}^n$.
\end{enumerate}
\end{lemma}
\begin{proof}
Point 1: Write $Q$ as $Q(x) = 2^\ell \cdot Q'(x)$, where $\ell < k$ is the largest power of $2$ that divides the GCD of the coefficients of $Q$. Projecting $Q'$ to a non-zero polynomial over $\Ringm{2}$ by dropping all its coefficients modulo $2$ and applying the standard Schwartz-Zippel lemma over $\Ringm{2}$ completes the proof.

Point 2 follows from point 1, and point 4 from point 3.
\end{proof}

\begin{lemma}
\label{fac:gamma}
Fix any $F\in \mc{F}_n$.
\begin{enumerate}
\item $\gamma_{d,k}(F) \geq \frac{1}{2}$.
\item  $\gamma_{d,k+1}(F)\geq \gamma_{d,k}(F)$. 
\item $\gamma_{d,k}(F) > 1-\frac{1}{2^d} \Rightarrow \gamma_{d,k}(F) = \gamma_{d,1}(F)$. 
\item $\gamma_{1,k}(F) = \gamma_{1,1}(F)$. 
\end{enumerate}
\end{lemma}
\begin{proof}
Point 1 is trivial since there is a constant polynomial that has agreement at least $\frac{1}{2}$ with $F_k$. 

Point $2$: Say $P\in \mc{P}_{d,k}$ has agreement $\alpha$ with $F_k$. Then, $2\cdot P$ (interpreted naturally as a polynomial in $\mc{P}_{d,k+1}$) has agreement $\alpha$ with $F_{k+1}$. 

For point $3$, consider a polynomial $Q\in \mc{P}_{d,k}$ that achieves the maximum agreement $\alpha > 1-\frac{1}{2^d}$ with $F_k$. Let $Q'\in \mc{P}_{d,1}$ be the polynomial obtained from $Q$ by dropping all its co-efficients modulo $2^{k-1}$. Note that for any $x$, $Q(x)\in \{0,2^{k-1}\}$ implies that $Q'(x) = 0$ (in the ring $\Ringm{2^{k-1}}$). Hence, the probability that $Q'$ is zero is at least $\alpha > 1-\frac{1}{2^d}$.~\cref{fac:polys} point 1 implies that $Q'$ must be the zero polynomial. Equivalently, all of the coefficients of $Q$ are divisible by $2^{k-1}$ and hence $Q$ can be naturally identified with $2^{k-1}\cdot Q''$ for some $Q''\in \Ringm{2}[x_1,\ldots,x_n]$. It is easy to check that $\agr(Q'',F_1) = \alpha$ and hence we have $\gamma_{d,1}(F) \geq \gamma_{d,k}(F)$. On the other hand, from point $2$, we already know that $\gamma_{d,k}(F) \leq \gamma_{d,1}(F)$. Hence we are done.

Point $4$ follows from points $1$ and $3$.

\end{proof}

\section{Some separation results}
\label{sec:examples}

\subsection{Symmetric functions as separating examples}
\label{subsec:sym}
We know from \cref{thm:alonbeig} that, for every fixed $\ell \ge 2$, $\gamma_{2^\ell-1,1}(S_{2^\ell}) \le \frac{1}{2} + o(1)$. In contrast, the main result of this section shows that
\begin{theorem}
 \label{theorem:elem-agr}
 For every fixed $\ell \ge 2$, $\gamma_{d,3}(S_{2^\ell}) \ge \frac{9}{16} - o(1)$, where $d = 2^{\ell-1} + 2^{\ell-2}$.
\end{theorem}
Notice that $2^{\ell-1} + 2^{\ell-2} \le 2^\ell -1$ for $\ell \ge 2$. This implies that, for $\ell \ge 2$, $S_{2^\ell}(x)$ is an example of a function $F$ for which there exist $k,d \in \mathbb{N}$ such that $\gamma_{d,1}(F) \le \frac{1}{2} + o(1)$ but $\gamma_{d',k}(F) \ge \frac{1}{2} + \Omega(1)$ for some $d' \le d$.
\begin{proof}[Proof of \cref{theorem:elem-agr}]
~\cref{fac:elemsym1} from \cref{sec:prelims} tells us that $S_{2^\ell}(x) = |x|_\ell$.
Thus, $S_{2^\ell,3}(x) \in \Ringm{8}[x_1, \ldots, x_n]$, the $3$-lift of $S_{2^\ell}(x)$, is given by
\begin{equation}
\label{eqn:lift}
 S_{2^\ell,3}(x) = 
 \begin{cases} 
 4 &\mbox{if } |x|_\ell = 1\\
 0 &\mbox{otherwise}
 \end{cases}
\end{equation}

Fix $d$ to be $2^{\ell-1} + 2^{\ell-2}$. Consider the polynomial $P(x) = \sum_{T \in {[n] \choose d}} \prod_{i \in T} x_i$ in $\Ringm{8}[x_1, \ldots, x_n]$.
To prove the theorem, it suffices to show that
$$ \Pr_{x \sim \{0,1\}^n }[P(x) = S_{2^\ell,3}(x)] \ge \frac{1}{2} + \frac{1}{16} - o(1).$$

Clearly, $P(x) = \dbinom{\lvert x\rvert}{d} \mbox{ mod } 8$, and
\begin{equation}
\label{eqn:def}
 P(x) = \begin{cases}
         0  & \mbox{if } 8 \mid \dbinom{|x|}{d}\\
         4  & \mbox{if } 4 \mid \dbinom{|x|}{d} \mbox{ but } 8 \nmid \dbinom{|x|}{d}\\
        \end{cases}
\end{equation}

The following theorem due to Kummer (see, e.g.,~\cite[Section 1.2.6, Ex. 11]{Knuth-tacopI}) determines the largest power of a prime that divides a binomial coefficient.
\begin{theorem}[Kummer]
 \label{thm:kummer}
 Let $p$ be a prime and $N,M \in \mathbb{N}$ such that $N \ge M$. Suppose $r$ is the largest integer such that $p^r \mid {N \choose M}$. Then $r$ is equal to the number of borrows required when subtracting $M$ from $N$ in base $p$.
\end{theorem}

Let $B(x)$ be the number of borrows required when subtracting $d$ from $|x|$. Rewriting~\cref{eqn:def} in terms of $B(x)$ using Kummer's theorem, we get
\begin{equation}
\label{eqn:poly}
 P(x) = \begin{cases}
         4  & \mbox{if } B(x) =2\\
         0  & \mbox{if } B(x) \ge 3
        \end{cases}
\end{equation}
We will need the following lemma.
\begin{lemma}\label{lem:main}
$P(x) = S_{2^\ell,3}(x)$ if either
\begin{enumerate}
\item $|x|_{\ell-2} = 0$, or
\item $(|x|_{\ell-2},|x|_{\ell-1},|x|_{\ell},|x|_{\ell+1}) = (1,0,0,0)$.
\end{enumerate}
\end{lemma}
\begin{proof}
Since $d = 2^{\ell-1} + 2^{\ell-2}$, all the bits of $d$ except $d_{\ell-1}$ and $d_{\ell-2}$ are zero. Thus, when subtracting $d$ from $|x|$, no borrows are required by the bits $|x|_i$, $0 \le i \le \ell-3$.

Using the above observation, the reader can verify that when $(|x|_{\ell-2},|x|_{\ell-1},|x|_{\ell},|x|_{\ell+1}) = (1,0,0,0)$ the number of borrows required is at least $3$ i.e. $B(x) \ge 3$, which in turn implies that $P(x) = 0$. Since $|x|_\ell = 0$, $S_{2^\ell,3}(x) = 0$. This proves the second part of the lemma.

To prove the first part, suppose $|x|_{\ell-2} = 0$. Since $d_{\ell-1}=d_{\ell-2}=1$, it follows that both $|x|_{\ell-2}$ and $|x|_{\ell-1}$ will need to borrow when subtracting $d$ from $|x|$. As argued before, no borrows are required by the bits before (i.e. less significant than) $|x|_{\ell-2}$, and thus the total number of borrows required by the bits $|x|_{i}$, ${0 \le i \le \ell-1}$, is 2.\\
Note that the bit $|x|_{\ell-1}$ borrows from $|x|_\ell$. Consider the following case analysis:
\begin{itemize}
\item Case $|x|_\ell = 1$: $|x|_\ell$ will not need to borrow since $d_\ell = 0$. In fact, none of the bits after (i.e. more significant than) $|x|_\ell$ will need to borrow, and thus $B(x) = 2$. This implies that $P(x) = 4$. We also have $S_{2^\ell,3}(x) = 4$ and hence $P(x) = S_{2^\ell,3}(x)$.
\item Case $|x|_\ell = 0$: $|x|_\ell$ will require a borrow and this means $B(x) \ge 3$. This implies that $P(x) = 0$. Since $|x|_\ell = 0$, it follows that $P(x) = S_{2^\ell,3}(x)$.
\end{itemize}
This completes the proof.
\end{proof}

By~\cref{lem:main}, we have
\begin{equation}
\label{eqn:main}
\Pr[P(x) = S_{2^\ell,3}(x)] \geq \Pr[ |x|_{\ell-2} = 0] + \Pr\left[(|x|_{\ell-2},|x|_{\ell-1},|x|_{\ell},|x|_{\ell+1}) = (1,0,0,0)\right]
\end{equation}
Using~\cref{fac:folk} from \cref{sec:prelims}, we have
\begin{eqnarray*}
&&\Pr[ |x|_{\ell-2} = 0] \ge \frac{1}{2} - o(1)\\
&&\Pr[(|x|_{\ell-2},|x|_{\ell-1},|x|_{\ell},|x|_{\ell+1}) = (1,0,0,0)] \ge \frac{1}{16} - o(1)
\end{eqnarray*}
which, together with~\cref{eqn:main}, implies
\begin{equation}
 \Pr[P(x) = S_{2^\ell,3}(x)] \ge \frac{1}{2} + \frac{1}{16} - o(1). \nonumber
\end{equation}

\end{proof}

\subsection{A separation at $k=2$}
\label{sec:quadex}

Let $d\in \mathbb{N}$ be any power of $2$. In this section, we show that there are functions $F$ for which $\gamma_{2d,2}(F) > \gamma_{3d-1,1}(F)$.

\begin{theorem}
\label{thm:quadex}
For large enough $n$, there exists a function $F\in \mc{F}_{2n}$ such that $\gamma_{2d,2}(F) \geq \frac{3}{4}-o(1)$ but $\gamma_{3d-1,1}(F) \leq \frac{5}{8}+o(1)$.
\end{theorem}
In particular, we see that $\gamma_{2,2}(F) > \gamma_{2,1}(F)$. This result is notable, since it shows that there is a separation at the first place where it is possible to have one (Recall that  $\gamma_{1,k}(F) = \gamma_{1,1}(F)$ for any $F \in \mathcal{F}_n$ by~\cref{fac:gamma}).

%

Let us begin the proof of \cref{thm:quadex}. We first define a family of Boolean functions on $\{0,1\}^{2n}$. We denote the $2n$ variables by $x_1,\ldots,x_n$ and $y_1,\ldots,y_n$. We use $\binom{|x|}{d}$ to denote the $d$th elementary symmetric polynomial from the ring $\mathbb{Z}/4\mathbb{Z}[x_1,\ldots,x_n]$, i.e., $\binom{|x|}{d} = \sum_{S\in \binom{[n]}{d}}\prod_{i\in S}x_i$.\footnote{We distinguish between $\binom{|x|}{d}$ and $S_d(x)$ since the former is from $\mathbb{Z}/4\mathbb{Z}[x_1,\ldots,x_n]$ and latter a polynomial in $\F_2[x_1,\ldots,x_n]$.} 

We will need the following easy corollary of~\cref{thm:kummer}.
\begin{corollary}
\label{cor:kummer}
Let $d$ be a power of $2$. Then, for $N\geq d$, the highest power of $2$ dividing $\binom{N}{d}$ is equal to the highest power of $2$ dividing $\lfloor \frac{N}{d}\rfloor$. 
\end{corollary}

Let $S = \{(x,y)\ |\ \binom{|x|}{d},\binom{|y|}{d} \equiv 1 \pmod{2}\}$. Given any function $H: \{0,1\}^{2n}\rightarrow \{0,1\}$, we define the Boolean function $F_H(x_1,\dots,x_n,y_1,\ldots,y_n)$ as follows: 
\[
F_H(x,y) = \left\{
\begin{array}{cc}
0 & \text{if $\binom{|x|}{d}\cdot\binom{|y|}{d} \equiv 0 \pmod{4}$},\\
1 & \text{if $\binom{|x|}{d}\cdot\binom{|y|}{d} \equiv 2 \pmod{4}$},\\
H(x,y) & \text{otherwise.}
\end{array}
\right.
\]

Define $P(x,y) = \binom{|x|}{d}\cdot \binom{|y|}{d}\in \mathbb{Z}/4\mathbb{Z}[x_1,\ldots,x_n,y_1,\ldots,y_n]$. Note that $F_H(x,y)$ is defined so that its $2$-lift agrees with $P(x,y)$ on points $(x,y)$ where $P(x,y)\in \{0,2\}$. Also~\cref{cor:kummer} implies that the following is an alternate equivalent definition of $F_H$ in terms of elementary symmetric polynomials modulo $2$.
\begin{equation}
\label{eq:altF}
F_H(x,y) = \left\{ 
\begin{array}{ll}
0 & \text{if $S_d(x) = S_d(y) = 0$,}\\
S_{2d}(y) & \text{if $S_d(x) = 1$ and $S_d(y) = 0$,}\\
S_{2d}(x) & \text{if $S_d(x) = 0$ and $S_d(y) = 1$,}\\
H(x,y) & \text{otherwise.}
\end{array}
\right.
\end{equation}


We now begin the proof of \cref{thm:quadex}. First of all, let us note that for any choice of $H$, we have:

\begin{lemma}
\label{lem:gamma22ubd}
$\gamma_{2d,2}(F_H) \geq \frac{3}{4}-o(1)$.
\end{lemma}

\begin{proof}
Consider the polynomial $P(x,y)\in \mc{P}_{2d,2}$ defined above. From~\cref{eq:altF}, it follows that the probability that $P(x,y) \neq F_{H,2}(x,y)$\footnote{$F_{H,2}$ denotes the $2$-lift of $F_H$.} is less than or equal to the probability that $S_d(x) = S_d(y) = 1$, which is $\frac{1}{4} + o(1)$ by~\cref{cor:folk}. This gives the claim.
\end{proof}

The main lemma is the following. 

\begin{lemma}
\label{lem:gamma21lbd}
Say $H:\{0,1\}^{2n}\rightarrow \{0,1\}$ is chosen uniformly at random. Then, $$\prob{H}{\gamma_{3d-1,1}(F_H) > \frac{5}{8} + o(1)} = o(1).$$
\end{lemma}
This will prove \cref{thm:quadex}. We will prove the above lemma in the following subsection.

\subsection{Proof of {\cref{lem:gamma21lbd}}}

The outline of the proof is as follows. Fix any polynomial $Q\in \F_2[x_1,\ldots,x_n,y_1,\ldots,y_n]$ of degree at most $3d-1$. We need to show that $\agr(F_H,Q)\leq \frac{5}{8}+o(1)$ for a random $H:\{0,1\}^{2n}\rightarrow\{0,1\}$. The fact that $H$ is random ensures that any $Q$ cannot agree with $H$ on significantly more than half the inputs in $S$. For inputs outside $S$, we need a more involved argument, following Alon and Beigel~\cite{AlonB2001}. We show that for any $Q$ we can find somewhat large sets $I$ and $J$ of $x$ and $y$ variables respectively such that when we set the variables outside $I\cup J$, we obtain a polynomial that is symmetric in the variables of $I\cup J$. This is a Ramsey theoretic argument \'{a}la Alon-Beigel~\cite{AlonB2001}. 

Following this argument, we only need to prove the agreement upper bound for $Q$ that is symmetric in $x$ and $y$ variables. This can be done by reduction to a constant-sized problem, as we show below. A careful computation to solve the constant-sized problem will finish the proof.

We begin with some notation that will be useful in the proof. Throughout, we work with disjoint sets of $x$-variables and $y$-variables of equal size and consider polynomials over these variables. Let the $x$-variables be $\{x_1,\ldots,x_n\}$ and the $y$-variables be $\{y_1,\ldots,y_n\}$. For $I,J\subseteq [n]$, the set of $\F_2$-polynomials $Q$ over the variables $\{x_i\ |\ i\in I\}$ and $\{y_j\ |\ j\in J\}$ is denoted $\F_2[x_I,y_J]$. Similarly, we use $Q\in \F_2[x_I]$ to denote the fact that $Q$ is a polynomial only over the variables $\{x_i\ |\ i\in I\}$.

We use $\mc{A}_{I,J}$ to denote Boolean assignments $\sigma:\{x_i\ |\ i\not\in I\}\cup \{y_j\ |\ j\not\in J\}\rightarrow \{0,1\}$. Given $F:\F_2^{2n}\rightarrow \F_2$ and $\sigma \in \mc{A}_{I,J}$, we use $F|_{\sigma}\in \F_2[x_I,y_J]$ to denote its natural restriction to the variables indexed by $I\cup J$.

We say that $Q\in \F_2[x_I,y_J]$ is \emph{$(x,y)$-symmetric} if it is a linear combination of the polynomials in the set $\{S_{d_1}(x_I)\cdot S_{d_2}(y_J)\ |\ d_1,d_2\in \mathbb{N}\}$. We note that being $(x,y)$-symmetric depends on the sets $I,J$ under consideration. This will be implicit when used.

Given a multilinear monomial $m$ over the $x$ and $y$-variables, its \emph{multidegree} is defined to be $(i,j)$ if $m$ multiplies $i$ $x$-variables and $j$ $y$-variables. Let $\mc{D} = \{(i,j)\ |\ i+j\leq 3d-1\}$ be the set of multidegrees of monomials of degree at most $3d-1$. We order $\mc{D}$ in ascending order according to $i+j$, i.e., fix a total ordering $\preceq$ of $\mc{D}$ such that if $i_1+j_1 < i_2 + j_2$, then $(i_1,j_1)\preceq (i_2,j_2)$ \footnote{If $(i_1,j_1) \neq (i_2, j_2)$, but $i_1 + j_1 = i_2 + j_2$, then the relation between $(i_1,j_1)$ and $(i_2, j_2)$ is fixed in an arbitrary manner.}. Let $(i_0,j_0)$ be the largest element in the ordering $\preceq$. We will define $\mathrm{multdeg}(Q)$ to be the largest (w.r.t. $\preceq$) multidegree of a monomial that has a non-zero coefficient in $Q$. 
For $(i,j)$ such that $i+j\leq 3d-1$, we say that a polynomial $Q\in \F_2[x_I,y_J]$ is $(x,y;i,j)$-symmetric if we can write $Q$ as 
\begin{equation}
\label{eq:xyijsym}
Q = Q_1 \oplus Q_2
\end{equation}
where $\mathrm{multdeg}(Q_1)\preceq (i,j)$ and $Q_2$ is $(x,y)$-symmetric. Note that if $(i,j) = (i_0,j_0)$, then any polynomial of degree at most $3d-1$ is $(x,y;i_0,j_0)$-symmetric, since we can take $Q_1 = Q$ and $Q_2=0$.

We also need the following variant of the function $F_H$ defined above. Call a function $\Phi\in \F_2[x_1,\ldots,x_n]$ (resp.\ $\Psi\in \F_2[y_1,\ldots,y_n]$) $d$-simple w.r.t.\ $x$ (resp.\ w.r.t.\ $y$) if it is a linear combination of symmetric polynomials in $x$ (resp.\ in $y$) of degree \emph{strictly} less than $d$. Equivalently, we can say that $\Phi(x)$ only depends upon $|x|_0,\ldots,|x|_{\lg d -1}$, and similarly for $\Psi(y)$ w.r.t. $y$.

Given pairs of polynomials $\Phi = (\Phi_1,\Phi_2) \in \F_2[x_1,\ldots,x_n]\times \F_2[x_1,\ldots,x_n]$, and $\Psi = (\Psi_1,\Psi_2) \in \F_2[y_1,\ldots,y_n]\times \F_2[y_1,\ldots,y_n]$, such that $\Phi_1,\Phi_2$ and $\Psi_1,\Psi_2$ are $d$-simple w.r.t. $x$ and $y$ respectively, define
\begin{equation}
\label{eq:altFPhiPsi}
F_{H,\Phi,\Psi}(x,y) = \left\{ 
\begin{array}{ll}
0 & \text{if $S_d(x) =\Phi_1(x), S_d(y) = \Psi_1(y)$,}\\
S_{2d}(y) \oplus \Psi_2(y) & \text{if $S_d(x) = 1\oplus\Phi_1(x)$ and $S_d(y) = \Psi_1(y)$,}\\
S_{2d}(x) \oplus \Phi_2(x) & \text{if $S_d(x) = \Phi_1(x)$ and $S_d(y) = 1\oplus\Psi_1(y)$,}\\
H(x,y) & \text{otherwise.}
\end{array}
\right.
\end{equation}
Also, define $S_{\Phi,\Psi} = \{(x,y)\ |\ S_d(x) = 1\oplus\Phi_1(x), S_d(y) = 1\oplus\Psi_1(y)\}$.

With the notation above, we are ready to state a claim that generalizes \cref{lem:gamma21lbd}.

\begin{lemma}
\label{lem:stronger}
Fix any $d\in \mathbb{N}$. For $(i,j)\in \mc{D}$ and $\varepsilon \in (0,1)$, there is an $n(i,j,\varepsilon)$ such that given any $n\geq n(i,j,\varepsilon)$, for uniformly random $H:\{0,1\}^{2n}\rightarrow \{0,1\}$, we have for any $\Phi=(\Phi_1,\Phi_2)\in \F_2[x_1,\ldots,x_n]^2$ and $\Psi=(\Psi_1,\Psi_2)\in \F_2[y_1,\ldots,y_n]^2$ such that $\Phi_1,\Phi_2$ are $d$-simple w.r.t.\ $x$, and $\Psi_1,\Psi_2$ are $d$-simple w.r.t.\ $y$,
\begin{equation}
\label{eq:stronger}
\prob{H}{\exists Q\text{ of degree $\leq 3d-1$ that is $(x,y;i,j)$-symmetric s.t. } \agr(F_{H,\Phi,\Psi},Q) \geq \frac{5}{8}+\varepsilon} \leq \varepsilon.
\end{equation}
\end{lemma}

The statement of the above lemma for $\Phi=\Psi = (0,0)$ and $(i,j) = (i_0,j_0)$ implies~\cref{lem:gamma21lbd} since in this case $F_{H,\Phi,\Psi} = F_H$ and as noted above, any polynomial $Q$ of degree at most $3d-1$ is $(x,y;i_0,j_0)$-symmetric.

The proof of~\cref{lem:stronger} is by induction on the order $\preceq$. The base case is the case when $(i,j) = (0,0)$, the minimal element of the ordering $\preceq$.

Throughout, the parameter $d$ is a fixed integer power of $2$.
\subsubsection{Base case of the induction: $i = j = 0$}

In this case, by \cref{eq:xyijsym}, it is clear that $Q$ is an $(x,y)$-symmetric polynomial. We show in this case that bounding $\agr(F_{H,\Phi,\Psi},Q)$ reduces (for most $H$) to bounding the correlations between functions on $5$ inputs. A simple computation solves this problem.

Fix any bits $\phi,\psi\in \{0,1\}^2$. Define the Boolean function $f$ on $5$ variables $a_1,b_1,a_2,b_2,z\in \F_2$ as follows. Notice the similarity to~\cref{eq:altF,eq:altFPhiPsi}.
\begin{equation}
\label{eq:defn-f}
f_{\phi,\psi}(a_1,a_2,b_1,b_2,z) = \left\{ 
\begin{array}{ll}
0 & \text{if $a_1=\phi_1,b_1=\psi_1$,}\\
b_2 \oplus \psi_2 & \text{if $a_1 = 1\oplus \phi_1$ and $b_1 = \psi_1$,}\\
a_2 \oplus \phi_2& \text{if $a_1=\phi_1$ and $b_1=1\oplus\psi_1$,}\\
z & \text{otherwise.}
\end{array}
\right.
\end{equation}

Call a polynomial $q\in \F_2[a_1,a_2,b_1,b_2,z]$ \emph{relevant} if $q$ is a linear combination of monomials from the set $\{1,a_1,a_2,b_1,b_2,a_1b_1\}$. Let $\mc{R}$ denote the set of relevant polynomials. Note that relevant polynomials do not involve the variable $z$.

Given $H,Q'\in \F_2[x_1,\ldots,x_n,y_1,\ldots,y_n]$ and $\Phi,\Psi$ as in the statement of~\cref{lem:stronger}, we define $\agr_{S}(H,Q')$ to be $\prob{(x,y)\in S_{\Phi,\Psi}}{H(x,y) = Q'(x,y)}.$

We say that $H:\{0,1\}^{2n}\rightarrow \{0,1\}$ is \emph{$\varepsilon$-hard} if for any $Q'\in \F_2[x_1,\ldots,x_n,y_1,\ldots,y_n]$ of degree at most $3d-1$, we have
$\left|\agr_S(H,Q')-\frac{1}{2}\right| \leq \varepsilon.$

We need the following property of a random $H:\{0,1\}^{2n}\rightarrow\{0,1\}$. 

\begin{lemma}
\label{lem:existsH}
For any $\varepsilon > 0$, there is an $n_0(\varepsilon)\in \mathbb{N}$ such that if $n\geq n_0(\varepsilon)$, then for $H:\{0,1\}^{2n}\rightarrow \{0,1\}$ chosen uniformly at random $\prob{H}{\text{$H$ not $\varepsilon$-hard}}\leq \varepsilon$.
\end{lemma}

\begin{proof}
  The proof is a trivial union bound. The number of polynomials $Q'$
  of degree at most $3d-1$ is at most $2^{(2n)^{3d}}$ (there are
  $\sum_{i=0}^{3d-1}\binom{2n}{i}\leq (2n)^{3d}$ many possible
  monomials each has $2$ possible coefficients). For each such $Q'$,
  the expected number of locations $x\in S_{\Phi,\Psi}$ where
  $H(x) \neq Q'(x)$ is $|S_{\Phi,\Psi}|/2$. By a Chernoff bound, the
  probability that this number is not in the range
  $[|S_{\Phi,\Psi}|/2 - \varepsilon|S_{\Phi,\Psi}|,|S_{\Phi,\Psi}|/2 +
  \varepsilon|S_{\Phi,\Psi}|]$
  is $\exp(-\Omega(\varepsilon^2 |S_{\Phi,\Psi}|)).$
  By~\cref{cor:folk}, it follows that
  $|S_{\Phi,\Psi}| = \Omega(2^{2n})$, and hence the above probability
  can be upper bounded by $\exp(-\Omega(\varepsilon^2 2^{2n}))$. A
  union bound over all the possible $Q'$ tells us that with probability
  $1-\exp((2n)^{3d}-\Omega(\varepsilon^2 2^{2n}))$ over the choice of
  $H$, every $Q'$ of degree at most $3d-1$ satisfies
  $\agr_S(H,Q') \in
  [\frac{1}{2}-\varepsilon,\frac{1}{2}+\varepsilon]$.
  In particular, for any $\varepsilon$, a large enough $n$ will ensure
  that the probability that $H$ is not $\varepsilon$-hard is at most
  $\varepsilon$.
\end{proof}

We will prove the following lemmas.

\begin{lemma}
\label{lem:rel-suff}
Fix any $\Phi,\Psi$ as in the statement of~\cref{lem:stronger}. For any $\varepsilon > 0$, there is an $n(0,0,\varepsilon)\in \mathbb{N}$ such that for any $n\geq n(0,0,\varepsilon)$
\[
\prob{H}{\exists Q\text{ of degree $\leq 3d-1$ that is $(x,y)$-symmetric s.t. } \agr(F_{H,\Phi,\Psi},Q) > \max_{q\in \mc{R},\phi,\psi}\agr(f_{\phi,\psi},q)+\varepsilon} \leq \varepsilon.
\]
\end{lemma}

\begin{lemma}
\label{lem:rel-lbd}
$\max_{q\in \mc{R},\phi,\psi}\agr(f_{\phi,\psi},q)\leq \frac{5}{8}$.
\end{lemma}

The above lemmas clearly prove \cref{eq:stronger} in the case $i = j = 0$, which completes the base case.

\begin{proof}[Proof of~\cref{lem:rel-suff}]
We choose $n(0,0,\varepsilon)$ during the course of the proof. First of all, we will assume that $n(0,0,\varepsilon)\geq n_0(\varepsilon/2)$, so that we have
\begin{equation}
\label{eq:Hhard}
\prob{H}{\text{$H$ not $\varepsilon/2$-hard}}\leq \varepsilon/2.
\end{equation}

We now show that when $H$ is $\varepsilon/2$-hard, then for any $Q$ that is $(x,y)$-symmetric of degree at most $3d-1$, we have 
\begin{equation}
\label{eq:ifHhard}
\agr(F_{H,\Phi,\Psi},Q) \leq \max_{q\in \mc{R},\phi,\psi}\agr(f_{\phi,\psi},q)+\varepsilon.
\end{equation}
This will prove the lemma. Fix any $\varepsilon/2$-hard $H$ for the remainder of the lemma. 

Since $Q$ is $(x,y)$-symmetric, it follows that we can write
\begin{equation}
\label{eq:Qsigma}
Q = \bigoplus_{d_1,d_2: d_1 + d_2 \leq 3d-1}\gamma_{d_1,d_2} S_{d_1}(x)\cdot S_{d_2}(y)
\end{equation}
for some choice of the $\gamma_{d_1,d_2}$s from $\F_2$.
%

Let $z$ be a new variable taking values in $\F_2$. We now define $F \in \F_2[x_1,\ldots,x_n,y_1,\ldots,y_n,z]$ as follows. We set $F(x,y) = F_{H,\Phi,\Psi}(x,y)$ for all $(x,y)\not\in S_{\Phi,\Psi}$ and $F(x,y) = z$ for all $(x,y)\in S_{\Phi,\Psi}$. Note that we have $\agr_S(F,Q) = \frac{1}{2}$ since $Q$ does not depend on the random variable $z$. Further, since $H$ is $\varepsilon/2$-hard, we know that $|\agr_S(F_{H,\Phi,\Psi},Q)- \frac{1}{2}|\leq \varepsilon/2$.

In particular, we see that $|\agr_S(F_{H,\Phi,\Psi},Q) - \agr_S(F,Q)| \leq \varepsilon/2$. Since $F$ and $F_{H,\Phi,\Psi}$ agree outside $S_{\Phi,\Psi}$, this implies that 
\begin{equation}
\label{eq:F_HvsF}
|\agr(F_{H,\Phi,\Psi},Q) - \agr(F,Q)|\leq \varepsilon/2.
\end{equation}
So to upper bound $\agr(F_{H,\Phi,\Psi},Q)$, we upper bound $\agr(F,Q)$. Assume $d = 2^\ell$.

Consider $Q$. By \cref{eq:Qsigma}, we have
\begin{equation}
\label{eq:Qexp}
Q(x,y) = \bigoplus_{d_1,d_2: d_1+d_2 \leq 3d-1} \gamma_{d_1,d_2} S_{d_1}(x)\cdot S_{d_2}(y) = \bigoplus_{\substack{A,B\subseteq \{0,\ldots,\ell+1\}:\\ \sum_{i\in A} 2^i + \sum_{j\in B}2^j \leq 3d-1}}\gamma_{A,B} \prod_{i\in A}|x|_i \cdot \prod_{j\in B}|y|_j
\end{equation}
where the $\gamma_{A,B}$s are in $\F_2$, $|x|_0,\ldots,|x|_{\ell+1}$ being the $\ell+2$ least significant bits of $|x|$ (and similarly for $y$) and we have used~\cref{fac:elemsym1} for the final equality above.  

Now consider $F$. By the definition of $F$ above, \cref{eq:altFPhiPsi}, and once again using~\cref{fac:elemsym1}, we have
\begin{equation}
\label{eq:altF'}
F(x,y,z) = \left\{ 
\begin{array}{ll}
0 & \text{if $|x|_{\ell} =\Phi_1(x), |y|_\ell = \Psi_1$(y),}\\
|y|_{\ell+1} \oplus \Psi_2(y)& \text{if $|x|_{\ell} = 1\oplus\Phi_1(x)$ and $|y|_{\ell} = \Psi_1(y)$,}\\
|x|_{\ell+1} \oplus\Phi_2(y) & \text{if $|x|_{\ell} = \Phi_1(x)$ and $|y|_\ell = 1\oplus\Psi_1(y)$,}\\
z & \text{otherwise,}
\end{array}
\right.
\end{equation}
where $\Phi_1,\Phi_2$, being $d$-simple, are functions of $|x|_0,\ldots,|x|_{\ell-1}$, and similarly, $\Psi_1,\Psi_2$ are functions of $|y|_0, \ldots |y|_{\ell -1}$.

Let $\alpha_0,\ldots,\alpha_{\ell+1},\beta_0,\ldots,\beta_{\ell+1}$ be $2(\ell+2)$ new variables. Define $q'\in \F_2[\alpha_0,\ldots,\alpha_{\ell+1},\beta_0,\ldots,\beta_{\ell+1}]$ by replacing $|x|_i$ by $\alpha_i$ and $|y|_j$ by $\beta_j$ in \cref{eq:Qexp} above. That is,
\begin{equation}
\label{eq:defnq'}
q'(\alpha_0,\ldots,\alpha_{\ell+1},\beta_0,\ldots,\beta_{\ell+1}) =  \bigoplus_{\substack{A,B\subseteq \{0,\ldots,\ell+1\}:\\ \sum_{i\in A} 2^i + \sum_{j\in B}2^j \leq 3d-1}}\gamma_{A,B} \prod_{i\in A}\alpha_i \cdot \prod_{j\in B}\beta_j
\end{equation}
and similarly define $f'\in \F_2[\alpha_0,\ldots,\alpha_{\ell+1},\beta_0,\ldots,\beta_{\ell+1},z]$ by replacing $|x|_i$ by $\alpha_i$ and $|y|_j$ by $\beta_j$ for each $i,j\in \{0,\ldots,\ell+1\}$ in the definition of $F$ above. We have
\begin{align*}
\agr(F,Q) &= \prob{x,y,z}{F(x,y,z) = Q(x,y)}\\
&= \prob{x,y,z}{f'(|x|_0,\ldots,|x|_{\ell+1},|y|_0,\ldots,|y|_{\ell+1},z) = q'(|x|_0,\ldots,|x|_{\ell+1},|y|_0,\ldots,|y|_{\ell+1})}
\end{align*}

By~\cref{fac:folk}, we know that if $n(0,0,\varepsilon)$ is large enough, then for uniformly random $x,y\in \F_2^n$, the tuples $(|x|_0,\ldots,|x|_{\ell+1})$ and $(|y|_0,\ldots,|y|_{\ell+1})$ are $\varepsilon/4$-close to the uniform distribution (in statistical distance) over $\F_2^{\ell+2}$. Note also that $x,y,z$ are mutually independent. From this, it easily follows that the final expression in the above display is $\frac{\varepsilon}{2}$-close to $\agr(f',q')$. Thus, we get
\begin{equation}
\label{eq:F'f'}
|\agr(F,Q)-\agr(f',q')| \leq \frac{\varepsilon}{2}.
\end{equation}

Therefore, we analyze $\agr(f',q')$. Conditioning on any setting $\tau$ of $\alpha_0,\ldots,\alpha_{\ell-1}, \beta_0, \ldots, \beta_{\ell-1}$, we see that the functions $\Phi_1,\Phi_2,\Psi_1,\Psi_2$ (being $d$-simple) are fixed to some constants in $\F_2$ and hence $f'$ simplifies to a polynomial $f''(\alpha_{\ell},\alpha_{\ell+1},\beta_\ell,\beta_{\ell+1})$. Similarly, $q'$ simplifies to some $q''(\alpha_{\ell},\alpha_{\ell+1},\beta_\ell,\beta_{\ell+1})$. Further, note that by the constraints on sets $A$ and $B$ in  \cref{eq:defnq'}, $q''$ must be a linear combination of monomials from the set $\{1,\alpha_\ell,\alpha_{\ell+1},\beta_\ell,\beta_{\ell+1},\alpha_\ell\beta_\ell\}$. Renaming variables $\alpha_\ell,\alpha_{\ell+1},\beta_\ell,\beta_{\ell+1}$ to $a_1,a_2,b_1,b_2$ respectively, we see that $\agr(f'',q'') \leq \max_{q\in \mc{R},\phi,\psi\in \F_2^2}\agr(f_{\phi,\psi},q)$. Since this is true of any $\tau$, the same upper bound holds for $\agr(f',q')$ as well. 

Combined with \cref{eq:F'f'} and \cref{eq:F_HvsF}, this yields $\agr(F_{H,\Phi,\Psi},Q)\leq \max_{q\in \mc{R},\phi,\psi}\agr(f_{\phi,\psi},q) + \varepsilon$. Since this is true for every $\varepsilon/2$-hard function $H$, and the probability that a random $H$ is $(\varepsilon/2)$-hard is at least $1-\varepsilon/2$, we are done.
\end{proof}

\begin{proof}[Proof of~\cref{lem:rel-lbd}]
We prove the statement by a simple case analysis. 

The first case is that the relevant polynomial $q\in \mc{R}$ depends on at least one among $\{a_2,b_2\}$. Without loss of generality, we assume that $q$ depends on $a_2$. Then, by the definition of $\mc{R}$, we can write $q = a_2 \oplus q'$ where $q'\in \F_2[a_1,b_1,b_2]$. Consider any setting of $(a_1,b_1,b_2,z)$ such that $(a_1,b_1)\in \{(\phi_1,\psi_1),(1\oplus\phi_1,\psi_1),(1\oplus\phi_1,1\oplus\psi_1)\}$. Under this restriction, $f_{\phi,\psi}$ is a constant function whereas $q$ is a non-constant linear function depending on $a_2$. Hence, when $(a_1,b_1)\in \{(\phi_1,\psi_1),(1\oplus\phi_1,\psi_1),(1\oplus\phi_1,1\oplus\psi_1)\}$, $f_{\phi,\psi}$ and $q$ can agree on at most half the inputs. Thus we get that $\agr(f_{\phi,\psi},q)\leq \frac{5}{8}$.

The second case is that $q$ depends on neither $a_2$ nor $b_2$. In this case, consider any setting of $(a_1,b_1)\in \{(\phi_1,1\oplus\psi_1),(1\oplus\phi_1,\psi_1),(1\oplus\phi_1,1\oplus\psi_1)\}$. Under each of these restrictions, $q$ computes the constant function (recall that $q$ does not depend on $z$) whereas $f_{\phi,\psi}$ is a non-constant linear function. Thus, as before, we get that $\agr(f_{\phi,\psi},q)\leq \frac{5}{8}$. This proves the lemma.
\end{proof}

\subsubsection{The induction case}

We now induct. Let $(i,j)\in \mc{D}$ be non-minimal and let $(i',j')$ be its predecessor w.r.t. $\preceq$. Assume~\cref{lem:stronger} for $(x,y;i',j')$-symmetric polynomials. We now prove it for $(x,y;i,j)$-symmetric polynomials.

We will need the following basic Ramsey-theoretic statement. It is a
straightforward generalization (to hypergraphs) of the fact that any
large enough bipartite graph contains large bipartite independent sets
or complete bipartite subgraphs. Unfortunately we could not find
exactly this statement in the literature, so we provide a proof of the
statement in \cref{sec:ramsey}.

Let $I$ and $J$ be disjoint sets of size $n$ each. A function $c:\binom{I}{i}\times \binom{J}{j}\rightarrow \{0,1\}$ is said to be an $(i,j)$-colouring of $(I,J)$. (Recall that $\binom{I}{i}$ denotes the collection of all $i$-sized subsets of $I$.)
\begin{lemma}
\label{lem:ramsey-colour}
For any $i,j\in \mathbb{N}$ and any $r\in \mathbb{N}$, there is an $n_R(i,j,r)\geq 2^r\in \mathbb{N}$ such that for any $n \geq n_R(i,j,r)$, any disjoint $n$-sets $I,J$ and any $(i,j)$-colouring $c$ of $(I,J)$, there are sets $I'\subseteq I$ and $J'\subseteq J$ with $|I'| = |J'| = r$ such that the restriction $c'$ of $c$ to $\binom{I'}{i}\times \binom{J'}{j}$ is a constant function.
\end{lemma}

We now prove the inductive case of~\cref{lem:stronger}. Let $Q\in \F_2[x_1,\ldots,x_n,y_1,\ldots,y_n]$ be $(x,y;i,j)$-symmetric. By \cref{eq:xyijsym}, we have
\begin{equation}
\label{eq:xyijindn}
Q = Q_1 \oplus Q_2 = Q_1' \oplus Q_{i,j} \oplus Q_2
\end{equation}
where $Q_2$ is $(x,y)$-symmetric, $Q_1$ has multidegree at most $(i,j)$, $Q_{i,j}$ is the part of $Q_1$ of multidegree \emph{exactly} $(i,j)$, and $Q_1'$ is the part of multidegree strictly less than $(i,j)$ (i.e. at most $(i',j')$).

Use $Q_{i,j}$ to define an $(i,j)$-colouring $c$ of $([n],[n])$ as follows. For $A\in \binom{[n]}{i}$ and $B\in \binom{[n]}{j}$, we define $c(A,B)$ to be the coefficient of the monomial $\prod_{s\in A}x_s\cdot \prod_{t\in B}y_t$ in $Q_{i,j}$. Applying~\cref{lem:ramsey-colour} with $r = n(i',j',\varepsilon/4)$, we see that if $n\geq n_R(i,j,r)$, then there are $I,J\in \binom{[n]}{r}$ and $\alpha\in \F_2$ such that for all $A\in \binom{I}{i}$ and $B\in \binom{J}{j}$, we have $c(A,B) = \alpha$.

Assume that $Q$ is as in \cref{eq:xyijindn} and $n\geq n_R(i,j,r)$. We find $I,J$ as above. For any setting $\sigma\in \mc{A}_{I,J}$, we can write the polynomial $Q_{i,j}|_\sigma$ as $Q_{i,j}' \oplus Q_{i,j}''$ where $Q'_{i,j}$ is the part of degree $i+j$, and $Q_{i,j}''$ has degree strictly less than $i+j$.

Observe that 
\[
Q_{i,j}' = \bigoplus_{A\in \binom{I}{i}, B\in \binom{J}{j}} c(A,B)\prod_{s\in A}x_s \prod_{t\in B}y_t = \alpha \cdot S_i(x_{I})S_j(y_{J}),
\]
and is an $(x,y)$-symmetric polynomial (on the remaining variables $x_I,y_J$). Hence, by \cref{eq:xyijindn}, we get
\[
Q|_{\sigma} = Q_1'|_\sigma \oplus Q_{i,j}'' \oplus Q_{i,j}' \oplus Q_2|_{\sigma}.
\]
As observed above, $Q_{i,j}'$ is $(x,y)$-symmetric. Further, it is easily checked that any restriction of an $(x,y)$-symmetric polynomial continues to be $(x,y)$-symmetric on the remaining variables. Hence, $Q_2|_{\sigma}$ is also $(x,y)$-symmetric. Further, note that $Q_1'|_{\sigma}$ has multidegree at most $(i',j')$. Also, by definition, the degree of $Q_{i,j}''$ is strictly less than $i+j$ and hence the multidegree of $Q_{i,j}''$ is at most $(i',j')$. Altogether, this implies that $Q|_{\sigma}$ is a sum of an $(x,y)$-symmetric polynomial (i.e. $Q'_{i,j} \oplus Q_2|_\sigma$) and a polynomial of multidegree at most $(i',j')$ (i.e. $Q_1'|_\sigma \oplus Q''_{i,j})$. Thus, $Q|_\sigma$ is an $(x,y;i',j')$-symmetric polynomial on $r$ $x$-variables and $r$ $y$-variables, where $r = n(i',j',\varepsilon/4)$.


Now, we analyze $\agr(F_{H,\Phi,\Psi},Q)$. Note that choosing a random function $H:\{0,1\}^{2n}\rightarrow\{0,1\}$ is the same as choosing each of its restrictions $H|_\sigma:\{0,1\}^{2r}\rightarrow\{0,1\}$ independently and uniformly at random. 
%

We claim  that for each $\sigma$, by the induction hypothesis, we have
\begin{equation}
\label{eq:indnagr}
\prob{H|_\sigma}{\agr(F_{H,\Phi,\Psi}|_\sigma, Q|_\sigma)\geq \frac{5}{8}+\frac{\varepsilon}{4}}\leq \frac{\varepsilon}{4}.
\end{equation}

Assuming the above, we show how to finish the proof. Let $Y_Q$ denote the number of $\sigma$ such that $\agr(F_{H,\Phi,\Psi}|_\sigma, Q|_\sigma)\geq \frac{5}{8}+\frac{\varepsilon}{4}$. The random variable\footnote{Note that $Y_Q$ is a random variable since it depends on the random function $H$.} $Y_Q$ is a sum of $2^{2(n-r)}$ independent $0$-$1$ random variables with $\avg{}{Y_Q}\leq 2^{2(n-r)}\cdot \frac{\varepsilon}{4}$. Thus, by the Chernoff bound, we have
\begin{equation}
\label{eq:chernoff-indn}
\prob{H}{Y_Q\geq 2^{2(n-r)}\cdot \frac{\varepsilon}{2}} \leq \prob{H}{Y_Q-\avg{}{Y_Q}\geq 2^{2(n-r)}\cdot \frac{\varepsilon}{4}}\leq \exp(-\Omega(\varepsilon^2 2^{2(n-r)})).
\end{equation}

In the event that $Y_Q < 2^{2(n-r)}\cdot \frac{\varepsilon}{2}$, we have $\prob{\sigma}{\agr(F_{H,\Phi,\Psi}|_\sigma, Q|_\sigma)\geq \frac{5}{8}+\frac{\varepsilon}{4}}\leq \frac{\varepsilon}{2}$. Hence, we see that in this case
\begin{equation}
\label{eq:ifchernoff}
\agr(F_{H,\Phi,\Psi},Q) = \avg{\sigma}{\agr(F_{H,\Phi,\Psi}|_\sigma,Q|_\sigma)}\leq \left(\frac{5}{8}+\frac{\varepsilon}{4}\right) + \prob{\sigma}{\agr(F_{H,\Phi,\Psi}|_\sigma, Q|_\sigma)\geq \left(\frac{5}{8}+\frac{\varepsilon}{4}\right)}\leq \frac{5}{8} + \frac{3\varepsilon}{4}.
\end{equation}

In particular, the probability that there is any $Q\in \F_2[x_1,\ldots,x_n,y_1,\ldots,y_n]$ of degree at most $3d-1$ such that $\agr(F_{H,\Phi,\Psi},Q)\geq \frac{5}{8}+\varepsilon$ can be upper bound bounded, using \cref{eq:chernoff-indn} and a union bound over all such $Q$, by
\[
2^{(2n)^{3d-1}}\exp(-\Omega(\varepsilon^2 2^{n-r})) \leq \exp((2n)^{3d-1} - \Omega(\varepsilon^2 2^{n-r})) \leq \exp((2n)^{3d-1} - \Omega(\varepsilon^2 \frac{2^{n}}{n})) < \varepsilon.
\]
Here, we have used the fact that the number of polynomials $Q$ of degree at most $3d-1$ is equal to the number of ways of choosing the coefficients (in $\F_2$) of $\binom{2n}{0}+\cdots+\binom{2n}{3d-1}\leq (2n)^{3d-1}$ many monomials. The second inequality follows from the fact that $n\geq n_R(i,j,r)\geq 2^r$. The final inequality is true as long $n \geq n_1(\varepsilon)$ for some $n_1(\varepsilon)\in \mathbb{N}$.

Overall, we see that if we define $n(i,j,\varepsilon) = \max\{n_R(i,j,r), n_1(\varepsilon)\}$, then for any $n\geq n(i,j,\varepsilon)$, we have the statement of the lemma for $(x,y;i,j)$-symmetric polynomials. This completes the induction.

It remains to prove \cref{eq:indnagr}. Fix any $\sigma\in \mc{A}_{I,J}$. Let $F' = F_{H,\Phi,\Psi}|_\sigma$. We use $u$ and $v$ to denote assignments to the variables indexed by $I$ and $J$ respectively and $\tilde{u}$ and $\tilde{v}$ to denote their natural completions to an assignment to all the variables (i.e. the other variables are assigned by $\sigma$). Assume that $d = 2^\ell$.

By the definition of $F_{H,\Phi,\Psi}$ in \cref{eq:altFPhiPsi} and using~\cref{fac:elemsym1}, we have
\begin{equation}
\label{eq:altF'}
F'(u,v,z) = \left\{ 
\begin{array}{ll}
0 & \text{if $|\tilde{u}|_{\ell} = \Phi_1(\tilde{u}), |\tilde{v}|_\ell = \Psi_1(\tilde{v})$,}\\
|\tilde{v}|_{\ell+1} \oplus \Psi_2(\tilde{v}) & \text{if $|\tilde{u}|_{\ell} = 1\oplus\Phi_1(\tilde{u})$ and $|\tilde{v}|_{\ell} = \Psi_1(\tilde{v})$,}\\
|\tilde{u}|_{\ell+1} \oplus \Phi_2(\tilde{u}) & \text{if $|\tilde{u}|_{\ell} = \Phi_1(\tilde{u})$ and $|\tilde{v}|_\ell = 1\oplus\Psi_1(\tilde{v})$,}\\
H|_\sigma(u,v) & \text{otherwise,}
\end{array}
\right.
\end{equation}
where $\Phi_1,\Phi_2$, being $d$-simple, are functions of $|\tilde{u}|_0,\ldots,|\tilde{u}|_{\ell-1}$, and similarly, $\Psi_1,\Psi_2$ are functions of $|\tilde{v}|_0,\ldots,|\tilde{v}|_{\ell-1}$.

We would like to write the above in terms of the bits of $|u|$ and $|v|$. This is done as follows. Consider the case of $|\tilde{u}|_\ell$. Let $|\sigma_x|$ denote the number of $1$s assigned by $\sigma$ to the $x$ variables. Note that $|\tilde{u}| = |u| + |\sigma_x|$, and hence it follows that the function $\Phi_1(\tilde{u}) = \Phi_1|_\sigma(u)$ is a function of $|u|_0,\ldots,|u|_{\ell-1}$ and hence $d$-simple w.r.t.\ $u$; similarly, $\Psi_1|_\sigma(\tilde{v}) = \Psi_1|_\sigma(u)$ is $d$-simple w.r.t.\ $v$. Similarly, we can also write $|\tilde{u}|_{\ell} = |u|_\ell \oplus \Phi_1'(u)$ for some $d$-simple $\Phi_1'(u)$ depending on $\sigma$; also, $|\tilde{v}|_\ell = |v|_\ell \oplus \Psi_1'(v)$ for some $d$-simple $\Psi_1'(v)$ depending on $\sigma$.

Further elementary reasoning (left to the reader) allows us to deduce that there are $d$-simple $\Phi_2'(u),\Psi_2'(v)$ (depending on $\sigma$) such that 
\begin{align*}
|\tilde{u}|_{\ell+1} = |u|_{\ell+1} \oplus \Phi_2'(u) \text{  when $|\tilde{u}|_{\ell} = \Phi_1|_\sigma(u)$  and $|\tilde{v}|_{\ell} = 1\oplus\Psi_1|_\sigma(v)$}\\
|\tilde{v}|_{\ell+1} = |v|_{\ell+1} \oplus \Psi_2'(v) \text{  when $|\tilde{u}|_{\ell} = 1\oplus\Phi_1|_\sigma(u)$ and $|\tilde{v}|_{\ell} =\Psi_1|_\sigma(v)$}.
\end{align*}

The above along with \cref{eq:altF'} gives us
\begin{equation}
\label{eq:altaltF'}
F'(u,v,z) = \left\{ 
\begin{array}{ll}
0 & \text{if $|u|_{\ell} =\Phi_1''(u)$ and $|v|_\ell = \Psi_1''(v)$,}\\
|v|_{\ell+1} \oplus \Psi_2''(v) & \text{if $|u|_{\ell} = 1\oplus\Phi_1''(u)$ and $|v|_{\ell} = \Psi_1''(v)$,}\\
|u|_{\ell+1} \oplus \Phi_2''(u) & \text{if $|u|_{\ell} = \Phi_1''(u)$ and $|v|_\ell = 1\oplus\Psi_1''(v)$,}\\
H|_\sigma(u,v) & \text{otherwise,}
\end{array}
\right.
\end{equation}
where for each $i\in [2]$, $\Phi''_i(u)$ satisfies $\Phi''_i = \Phi_i \oplus \Phi'_i$, and is hence $d$-simple w.r.t.\ $u$, and similarly $\Psi''_i = \Psi_i \oplus \Psi'_i$ is $d$-simple w.r.t.\ $v$. Hence, we see that $F' = F_{H,\Phi,\Psi}|_\sigma = F_{H|_\sigma, \Phi'',\Psi''}$ (i.e. same as our hard function, but on $2r$ inputs). Using the fact that $r\geq n(i',j',\varepsilon/4)$, the induction hypothesis gives us 
\begin{align*}
\prob{H|_\sigma}{\exists Q''\in \F_2[x_I,y_J] \text{ of degree $\leq 3d-1$ and $(x,y;i',j')$-symmetric s.t. } \agr(F',Q'')\geq \frac{5}{8}+\frac{\varepsilon}{4}}\leq \frac{\varepsilon}{4}.
\end{align*}

In particular, since $Q|_\sigma$ is $(x,y;i',j')$-symmetric, we have $\prob{H|_\sigma}{\agr(F',Q|_\sigma)\geq \frac{5}{8}+\frac{\varepsilon}{4}}\leq \frac{\varepsilon}{4}$, which establishes \cref{eq:indnagr} and completes the proof.

\section{Upper bounds for $\gamma_{d,k}(\Maj_n)$}
\label{sec:smol}

In this section, we show an upper bound on $\gamma_{d,k}(\Maj_n)$ where $\Maj_n$ denotes the Majority function on $n$ bits\footnote{We define the majority function as $\Maj_n(x) = 1$ iff $|x| > n/2$.}

\begin{theorem}
\label{thm:smol-2powk}
For any $k\geq 1, d\in \mathbb{Z}^+$, $\gamma_{d,k}(\Maj_n) \leq \frac{1}{2} + \frac{10d}{\sqrt{n}}$.
\end{theorem}

The proof of~\cref{thm:smol-2powk} presented below is an adaptation of techniques appearing in a work of Green~\cite{Green2000}, who proved a similar result on the approximability of the parity function by polynomials over the ring $\Ringm{p^k}$, for prime $p \neq 2$.

We will need some definitions and facts about $\mc{P}_{d,k}$. 

We use $\pi$ to denote the unique ring homomorphism from $\Ringm{2^k}$ to $\Ringm{2}$. Its kernel $\pi^{-1}(0) = \{a\in \Ringm{2^k}\ |\ 2^{k-1}a = 0\}$ is the set of non-invertible elements in $\Ringm{2^k}$. 

We call a set $S\subseteq \{0,1\}^n$ \emph{forcing} for $\mc{P}_{d,k}$ if any polynomial $P\in \mc{P}_{d,k}$ that vanishes over $S$ is forced to take a value in $\pi^{-1}(0)$ at all points $x\in \{0,1\}^n$. Formally,
\[
(\forall x\in S\ \  P(x) = 0) \Rightarrow (\forall y \in \{0,1\}^n\ \  \pi(P(y))=0).
\]
Define the polynomial $\pi(P)\in \Ringm{2}[x_1,\ldots,x_n]$ to be the polynomial obtained by applying the map $\pi$ to each of the coefficients of $P$. Since a multilinear polynomial in $\Ringm{2^k}[x_1,\ldots,x_n]$ is the zero polynomial  iff it vanishes at all points of $\{0,1\}^n$ (by~\cref{fac:polys}), we see that $S$ is forcing iff $(\forall x\in S\ \  P(x) = 0) \Rightarrow \pi(P) = 0$.

Note that any interpolating set for $\mc{P}_{d,k}$ (see \cref{sec:prelims} for the definition) is forcing for $\mc{P}_{d,k}$, but the converse need not be true.

We now adapt the proof of Lemma 11 in \cite{Green2000} to bound the size of forcing sets for $\mc{P}_{d,k}$.
\begin{lemma}
\label{lem:min-forc}
If $S$ is forcing for $\mc{P}_{d,k}$, then $|S|\geq |\{0,1\}^n_{\leq d}| = \binom{n}{\leq d}$.
\end{lemma}
\begin{proof}
 Assume for the sake of contradiction that $S \subseteq \{0,1\}^n$ is forcing for $\mc{P}_{d,k}$ and $|S| < \binom{n}{\leq d}$. The latter implies the existence of a non-zero multilinear polynomial $Q(x) \in \mathbb{Q}[x_1, \ldots, x_n]$ of degree at most $d$ satisfying $Q(x) = 0$ for all $x \in S$. 

Let $Q'(x) = \sum_{T \subseteq [n]} c_T \prod_{i \in T}x_i $ be the polynomial in $\mathbb{Z}[x_1, \ldots, x_n]$ obtained by first clearing out the denominators of the coefficients of $Q(x)$, followed by dividing the resulting polynomial by the GCD of all the coefficients. Finally, let $P(x) = \sum_{T \subseteq [n]} c'_T \prod_{i \in T}x_i$ be any polynomial in $\mc{P}_{d,k}$ satisfying $c'_T \equiv c_T \Mod{2^k}$. It follows that $P$ is a non-zero polynomial of degree at most $d$ such that $\pi(P) \neq 0$, since $\pi(P) = 0$ would imply that every coefficent of $P$ (and thus every coefficent of $Q'$) is divisible by two, which is impossible since the coefficients of $Q'$ have no common divisor. 

To complete the proof, observe that $P(x) = 0$ for all $x \in S$, and since $S$ is forcing for $\mc{P}_{d,k}$, this implies that $\pi(P) = 0$, which is a contradiction. 

\end{proof}

We now use~\cref{lem:min-forc} to prove \cref{thm:smol-2powk}.

\begin{proof}[Proof of \cref{thm:smol-2powk}]
We assume throughout that $1\leq d \leq \frac{\sqrt{n}}{10}$; otherwise, there is nothing to prove. Let $\Maj_{n,k}:\{0,1\}^n \rightarrow \Ringm{2^k}$ be the $k$-lift of the $\Maj_n$ function. Let $P\in \mc{P}_{d,k}$ be arbitrary and let $S_P = \{x\in \{0,1\}^n\ |\ P(x) = \Maj_{n,k}(x)\}$. We want to show that $|S_P|\leq 2^n \cdot (\frac{1}{2} + \frac{10d }{\sqrt{n}})$. We will argue by contradiction. So assume that $|S_P| > 2^n \cdot (\frac{1}{2} + \frac{10 d}{\sqrt{n}})$.

Let $E_P$ be the complement of $S_P$, i.e. the set of points where $P$ makes an error in computing $\Maj_{n,k}$. We have $|E_P| < 2^n(\frac{1}{2} - \frac{10 d}{\sqrt{n}})$. We will try to find a degree $D$ (for suitable $D\leq \lfloor n/2 \rfloor $) polynomial $Q$ such that $Q$ vanishes at all points in $E_P$ but has the property that $Q(x)$ is a unit (i.e. $\pi(Q(x))\neq 0$) for some $x\in \{0,1\}^n$. To be able to do this, we need the fact that $E_P$ is not forcing for $\mc{P}_{D,k}$. By~\cref{lem:min-forc}, if $E_P$ is indeed forcing for $\mc{P}_{D,k}$, then 

\begin{align*}
|E_P|&\geq \sum_{i=0}^D \binom{n}{i} = \left(\sum_{i=0}^{\lfloor n/2 \rfloor} \binom{n}{i}\right) - \sum_{i = D+1}^{\lfloor n/2 \rfloor}\binom{n}{i}\\
 & \geq 2^{n-1} - \left(\lfloor n/2 \rfloor-D\right)\cdot \binom{n}{\lfloor n/2 \rfloor}\\ 
 &\geq 2^n\cdot \left(\frac{1}{2} - \frac{2(\lfloor n/2 \rfloor-D)}{\sqrt{n}}\right) = 2^n\cdot\left(\frac{1}{2} - \frac{4d}{\sqrt{n}}\right)
\end{align*}
where the last equality follows if we choose $D = \lfloor n/2 \rfloor-2d$. This contradicts our upper bound on the size of $|E_P|$. Hence, $E_P$ cannot be forcing for $\mc{P}_{D,k}$. In particular, we can find $Q$ that vanishes on $E_P$ and furthermore, $\pi(Q(x))\neq 0$ for some $x\in \{0,1\}^n$.

We now claim that $\pi(Q(x_0))\neq 0$ for some $x_0$ of Hamming weight $> n/2 $. To see this, consider the polynomial $Q_1 =\pi(Q)$. By construction of $Q$, we know that $Q_1$ is a non-zero polynomial of degree $D$. Hence, by~\cref{fac:polys}, $Q_1$ is non-zero when restricted to the Hamming ball of radius $D <  n/2$ around the all $1$s vector. In particular, this implies that there is an input $x_0$ of Hamming weight  $ > n/2$ where $Q_1(x_0)$ is non-zero and hence $\pi(Q(x_0))\neq 0$, or equivalently $2^{k-1}Q(x_0) \neq 0$. Fix this $x_0$ for the remainder of the proof. Note that $x_0\not\in E_P$ since $Q$ vanishes on $E_P$.

Now, consider the polynomial $R(x) = Q(x)\cdot P(x)$. We first show that $R(x) = 0$ for all $x$ of Hamming weight $\le n/2$. Consider any $x$ of Hamming weight $\le n/2$. If $x\in E_P$, then $R(x) = 0$ since $Q(x) = 0$. On the other hand, if $x\not\in E_P$, then $P(x) = \Maj_{n,k}(x) = 0$ since $x$ has Hamming weight $\le n/2$. Thus, $R$ vanishes at all inputs of Hamming weight $ \le n/2$. 

Since the degree of $R$ is at most $\deg(Q) + \deg(P) = D + d = (\lfloor n/2 \rfloor-2d) + d \le \lfloor n/2 \rfloor - d$ and $R$ vanishes at all inputs of $\{0,1\}^{n}_{\le n/2}$, this implies (by~\cref{fac:polys}) that $R$ must be $0$ everywhere. However, at $x_0$, $R(x_0) = Q(x_0)P(x_0) = Q(x_0)\Maj_{n,k}(x_0) = 2^{k-1}Q(x_0)\neq 0$. This yields the desired contradiction.
\end{proof}

\section{Connection to non-classical polynomials}\label{sec:nonclass}
Let $\mathbb{T} = \mathbb{R}/\mathbb{Z}$ denote the one dimensional torus. Observing that the additive structure of $\F_2$ is isomorphic to the additive subgroup $\{0,1/2\} < \mathbb{T}$, we can think of a Boolean function $F:\F_2^n \rightarrow \F_2$ as a function $F:\F_2^n \rightarrow \{0,1/2\}$, and conversely, a map $F:\F_2^n \rightarrow \{0,1/2\}$ as a Boolean function. 

Tao and Ziegler~\cite{TaoZ2012} give a characterization of non-classical polynomials as follows:
\begin{definition}[Tao and Ziegler~\cite{TaoZ2012}]\label{def:nonclass}
 A function $F:\F_2^n \rightarrow \mathbb{T}$ is a non-classical polynomial of degree $\le d$ if and only if it has the following form:
 $$F(x_1,\ldots,x_n) = \alpha + \sum_{0\le e_1, \ldots, e_n \le 1, k \ge 1: \sum_i e_i + (k-1) \le d} \frac{c_{e_1, \ldots, e_n,k} x_1^{e_1} \ldots x_n^{e_n}}{2^k} \Mod{1}$$
 Here  $\alpha \in \mathbb{T}$, and $c_{e_1,\ldots,e_n,k} \in \{0,1\}$ are uniquely determined. $\alpha$ is called the \textit{shift} of $F$, and the largest $k$ such that $c_{e_1, \ldots, e_n,k} \neq 0$ for some $(e_1, \ldots, e_n) \in \{0,1\}^n$ is called the \textit{depth} of $F$.
\end{definition}

Since we are interested in the agreement of a non-classical polynomial with Boolean ($\{0,1/2\}$-valued) functions, we will only consider polynomials with shift $\alpha = \frac{A}{2^{k}}$, where $k$ is the depth of the polynomial and $A\in \{0,\ldots,2^k-1\}$.  
\begin{remark}
\label{rem:class}
 Classical polynomials are non-classical polynomials with $\alpha \in \{0,1/2\}$ and depth $=1$. It is easy to see that every classical polynomial corresponds to a Boolean function. It is also not hard to show that every Boolean function can be represented as a classical polynomial. 
\end{remark}

The following lemma relates our model to non-classical polynomials:
\begin{lemma}
\label{lem:conn}
 Let $F$ be a Boolean function, and $d,k \in \mathbb{Z}^+$, $d \ge k$.
 \begin{enumerate}
  \item If there is a non-classical polynomial $P$ of degree $d$ and depth $k$ satisfying $\agr(F,P) = \gamma$, then there is a $P' \in \mathcal{P}_{d,k}$ satisfying $\agr(F_k,P') = \gamma$, where $F_k$ is the $k$-lift of $F$.
  \item If there is a $P \in \mathcal{P}_{d,k}$ satisfying $\agr(F_k,P) = \gamma$, then there is a non-classical polynomial $P'$ of degree $\le d + k -1$ and depth $k$ satisfying $\agr(F,P') = \gamma$.
 \end{enumerate}
 \end{lemma}
 \begin{proof}
  Fix $F$, $d$, and $k$ for the rest of the proof.

\textit{Proof of $1$}: Let $P$ be a non-classical polynomial of degree $d$ and depth $k$ with $\agr(F,P) = \gamma$. It is not hard to verify that $P$ can be written in the following form (See, e.g., proof of Lemma $2.2$ in \cite{BhowmickL2015}):
$$P(x) = \frac{P''(x)}{2^k} \Mod{1}$$
where $P''(x) \in \mathbb{Z}[x_1,\ldots,x_n]$ is of degree $d$. 

Suppose $P''(x) = \sum_{S \subseteq [n]} c_S \prod_{i \in S} x_i$. Choose $P' \in \mathcal{P}_{d,k}$, $P'(x) = \sum_{S \subseteq [n]} c'_S \prod_{i \in S} x_i$, satisfying
$$ \forall S \subseteq [n],\ c'_S \equiv c_S \Mod{2^k}.$$
By our choice of $P'$, we have that, for every $x \in \{0,1\}^n$ and $a \in \{0, \ldots, 2^k -1 \}$,
$$P(x) = \frac{a}{2^k} \Leftrightarrow P'(x) = a.$$
It follows that $\agr(F_k,P') = \gamma$.

\textit{Proof of $2$}: Let $P \in \mathcal{P}_{d,k}$ such that $\agr(F_k,P) = \gamma$. Using arguments similar to above, we can find a $P'' \in\mathbb{Z}[x_1,\ldots,x_n]$ of degree $d$ such that $P''(x) \equiv P(x) \Mod{2^k}$, for all $x \in \{0,1\}^n$. 

Define $P'(x)$ as
$$P'(x) = \frac{P''(x)}{2^k} \Mod{1}.$$
By comparing to the form in \cref{def:nonclass}, it is easy to see that $P'(x)$ is a non-classical polynomial of degree at most $d + k -1$ and depth $k$. Furthermore, we have that, for all $x \in \{0,1\}^n$ and $a \in \{0,\ldots,2^k -1\}$,
$$P(x) = a \Leftrightarrow P'(x) = \frac{a}{2^k}.$$
This completes the proof.

 \end{proof}

The first part of~\cref{lem:conn} implies the following corollary of \cref{thm:smol-2powk}:
\begin{corollary}
 Let $F: \F_2^n \rightarrow \mathbb{T}$ be a non-classical polynomial of degree $d$. Then,
 $$\Pr_{x \sim \F_2^n}[\maj_n(x) = F(x)] \le \frac{1}{2} + O\left(\frac{d}{\sqrt{n}}\right).$$
\end{corollary}
This proves a conjecture of Bhowmick and Lovett \cite{BhowmickL2015} that non-classical polynomials of degree $d$ do not approximate the Majority function any better than classical polynomials of the same degree.

The following is a consequence of \cref{thm:alonbeig} and the first part of~\cref{lem:conn}:
\begin{corollary}
\label{thm:alonbeig2}
 Let $\ell \ge 2$. Then, for every classical polynomial $P:\F_2^n \rightarrow \mathbb{T}$ of degree $\le 2^\ell - 1$,
 $$ \Pr_{x \sim \F_2^n} [P(x) = S_{2^\ell}(x)] \le \frac{1}{2} + o(1).$$
\end{corollary}
On the other hand, the second part of~\cref{lem:conn} and \cref{theorem:elem-agr} imply
\begin{corollary}
\label{cor:symm}
For every $\ell \ge 2$, there is a non-classical polynomial $F:\F_2^n \rightarrow \mathbb{T}$ of degree $\le 2^{\ell-1} + 2^{\ell-2} + 2$ and depth $3$ such that
$$ \Pr_{x \sim \F_2^n}[F(x) = S_{2^\ell}(x)] \ge \frac{9}{16} - o(1) $$
\end{corollary}
Noting that $2^{\ell -1 } + 2^{\ell -2} + 2 < 2^{\ell}$ for $\ell \ge 4$,~\cref{thm:alonbeig2} and~\cref{cor:symm} imply the following:
\begin{theorem}
 There is a Boolean function $F:\F_2^n \rightarrow \{0,1/2\}$ and $d \ge 1$, such that for every classical polynomial $P$ of degree at most $d$, we have
 $$\Pr_{x \sim \F_2^n}[F(x) = P(x)] \le \frac{1}{2} + o(1),$$
 but there is a non-classical polynomial $P'$ of degree $d' \le d$ satisfying
 $$\Pr_{x \sim \F_2^n}[F(x) = P'(x)] \ge \frac{1}{2} + \Omega(1).$$
\end{theorem}
This provides a counterexample to an informal conjecture of Bhowmick and Lovett~\cite{BhowmickL2015} that, for any Boolean function $F$, non-classical polynomials of degree $d$ do not approximate $F$ any better than classical polynomials of the same degree.

\section{Acknowledgements}
We would like to thank David Barrington for taking the time to explain Szegedy's~\cite{Szegedy1989} result to us, Arkadev Chattopadhyay for referring us to Green's result~\cite{Green2000}, and Swagato Sanyal for helpful discussions. We are also grateful to the anonymous reviewers for their detailed and helpful comments. In particular, we thank an anonymous reviewer for STACS 2017 who pointed out an error in the induction case of a previous proof of \cref{lem:stronger}. 

{\small
\bibliographystyle{prahladhurl}
\bibliography{BHS-bib}
}

\appendix

\section{Proof of {\cref{lem:ramsey-colour}}}\label{sec:ramsey}

\begin{proof}
Note that the constraint $n_R(i,j,r) \geq 2^r$ is easy to satisfy since if the latter part of the lemma holds for some $n_R(i,j,r) < 2^r$, then it continues to be the case for $n_R(i,j,r) = 2^r$. So we ignore the constraint $n_R(i,j,r) \geq 2^r$ for the rest of the proof. 

We prove by induction the following stronger statement. For any $i,j\in \mathbb{N}$ and any $r_0,s_0,r_1,s_1\in \mathbb{N}$, there is an $m_R(i,j;r_0,s_0,r_1,s_1)\in \mathbb{N}$ such that for any $n \geq m_R(i,j;r_0,s_0,r_1,s_1)$, any disjoint $n$-sets $I,J$ and any $(i,j)$-colouring $c$ of $(I,J)$, one of the following holds.
\begin{itemize}
\item There are sets $I'\subseteq I$ and $J'\subseteq J$ with $|I'| = r_0$ and $|J'| = s_0$ such that the restriction $c_1$ of $c$ to $\binom{I'}{i}\times \binom{J'}{j}$ is the constant $0$ function.
\item There are sets $I'\subseteq I$ and $J'\subseteq J$ with $|I'| = r_1$ and $|J'|= s_1$ such that the restriction $c_1$ of $c$ to $\binom{I'}{i}\times \binom{J'}{j}$ is the constant $1$ function.
\end{itemize}

Setting $r_0=r_1 = s_0=s_1=r$ above clearly yields the lemma.

The proof is by induction on $\min\{i,j\}$. Note that the statement is trivial when $i = j = 0$, since a $(0,0)$-colouring is by definition a constant function. So we can take $m_R(0,0;r_0, s_0, r_1, s_1) = \max\{r_0,s_0,r_1,s_1\}$ for any $r_0,r_1,s_0,s_1\in \mathbb{N}$. 

Now consider the case when $\min\{i,j\} = 0$ and $\max\{i,j\}\geq 1$; w.l.o.g. assume $j = \max\{i,j\}$. In this case, the function $c$ is essentially a colouring of $\binom{J}{j}$ and hence the statement  of the lemma reduces to the case of the standard Ramsey theorem for $j$-uniform hypergraphs. Thus, we know that $m_R(i,j;r_0,s_0,r_1,s_1)$ exists in this case. This completes the base case.

For the induction, assume the statement for any $r_0,s_0,r_1,s_1\in \mathbb{N}$ and any $(i',j')$ with $\min\{i',j'\} < k$ for some $k\geq 1$. Consider the case of $(i,j)$ such that $\min\{i,j\} = k$. Assume w.l.o.g. that $i = \min\{i,j\} \geq 1$. We now proceed by induction on $t = r_0+s_0+r_1+s_1$. 

The base case of the induction is when $\min\{r_0,s_0,r_1,s_1\} =0$, which is trivial as $i,j\geq 1$. For the induction case, assume that $\min\{r_0,s_0,r_1,s_1\} \geq 1$ and we have the statement for smaller values of $t$. W.l.o.g. assume that $r_0 = \min\{r_0,s_0,r_1,s_1\} \geq 1$.

By the induction hypotheses, we know the existence of $$m_1 = \max\{m_R(i,j;r_0-1,s_0,r_1,s_1),m_R(i,j;r_0,s_0,r_1-1,s_1)\}$$ and $m_2 = m_R(i-1,j;m_1,m_1,m_1,m_1)$. We claim that $m_R(i,j;r_0,s_0,r_1,s_1) = m_2 + 1$ has the required properties. 

To see this, consider any $(i,j)$-colouring $c$ of $(I,J)$ with $|I| = |J| \geq m_2 + 1$. Fix an arbitrary $a\in I$ and $b\in J$. Note that for $I_1 = I\setminus \{a\}$ and $J_1 = J\setminus \{b\}$, we obtain a $(i-1,j)$ colouring $c_a$ of $(I_1,J_1)$ by setting $c_a(A,B) = c(A\cup \{a\},B)$. Since $|I_1| = |J_1|\geq m_2$, we know that there exist $I_2\subseteq I_1$ and $J_2\subseteq J_1$ of size $m_1$ each such that the restriction of $c_a$ to $\binom{I_2}{i-1}\times \binom{J_2}{j}$ is a constant. Equivalently, there is an $\alpha \in \F_2$ such that for each $A\in \binom{I_2}{i-1}$ and $B\in \binom{J_2}{j}$, we have 
$c(A\cup \{a\},B) = \alpha.$

Assume $\alpha = 0$. Now, consider the restriction $c_2$ of $c$ to $\binom{I_2}{i}\times \binom{J_2}{j}$. Since $|I_2| = |J_2| = m_1 \geq m_R(i,j;r_0-1,s_0,r_1,s_1)$, we see that there exist $I''\subseteq I_2$ and $J''\subseteq J_2$ satisfying one of the following.
\begin{itemize}
\item $|I''| = r_0-1$ and $|J''| = s_0$, and the restriction $c''$ of $c$ to $\binom{I''}{i}\times \binom{J'}{j}$ is the constant $0$ function.
\item $|I''| = r_1$ and $|J''| = s_1$, and the restriction $c''$ of $c$ to $\binom{I''}{i}\times \binom{J'}{j}$ is the constant $1$ function.
\end{itemize}

In the former case, we can take $I' = I''\cup \{a\}$ and $J' = J''$ to prove the inductive case. Note that since $c_a$ is the constant $0$ function on $\binom{I''}{i-1}\times \binom{J'}{j}$, the restriction $c'$ of $c$ to $\binom{I'}{i}\times \binom{J'}{j}$ is also the constant $0$ function.

In the latter case, we just take $I' = I''$ and $J' = J''$, since we are guaranteed that the restriction of $c$ to $\binom{I'}{i}\times \binom{J'}{j}$ is the constant $1$ function.

In the case that $\alpha = 1$, we repeat the same proof except that we use the fact that $m_1 \geq m_R(i,j;r_0,s_0,r_1-1,s_1)$ to prove that there exist $I''\subseteq I_2$ and $J''\subseteq J_2$ satisfying one of the following.
\begin{itemize}
\item $|I''| = r_0$ and $|J''| = s_0$, and the restriction $c''$ of $c$ to $\binom{I''}{i}\times \binom{J'}{j}$ is the constant $0$ function.
\item $|I''| = r_1-1$ and $|J''| = s_1$, and the restriction $c''$ of $c$ to $\binom{I''}{i}\times \binom{J'}{j}$ is the constant $1$ function.
\end{itemize}

This proves the inductive case, and hence completes the proof.
\end{proof}

\end{document}